\documentclass[12pt,dvipsnames]{article}

\usepackage{verbatim, booktabs, multirow, hyperref, multirow}
\usepackage{graphicx}
\usepackage[round]{natbib}
\usepackage[utf8]{inputenc}
\usepackage{enumerate}
\usepackage{dirtytalk}
\usepackage{euscript}
\usepackage{amsmath, amsthm,amsfonts, amssymb}
\usepackage[margin = 1in]{geometry}
\usepackage{tikz}
\usepackage{bm}
\usepackage{pifont}
\usetikzlibrary{calc}
\usetikzlibrary{positioning}
\usepackage[normalem]{ulem}
\usepackage{accents}

\usepackage{algorithm}
\usepackage{algpseudocode}

\newtheorem{thm}{Theorem}

\newtheorem{ax}{Axiom}

\newtheorem{defn}{Definition}

\newtheorem*{example*}{Example}

\DeclareRobustCommand{\ubar}[1]{\underaccent{\bar}{#1}}
\newcommand{\argmax}{\operatornamewithlimits{argmax}}

\newcommand{\lmax}[2]{\argmax(#1,#2)}

\newcommand{\R}{\mathbb{R}}

\newcommand{\gB}{\mathcal{B}} 
\newcommand{\cM}{\mathcal{M}}

\newcommand{\C}{\bm{c}}

\title{Revealed Preference Analysis Under Limited Attention\thanks{
We would like to thank Yuta Inoue and Koji Shirai for providing the data set used in the paper.
We would also like to thank Joshua Lanier, Yusufcan Masatlioglu, and Koji Shirai for providing valuable comments.}}

\author{Mikhail Freer\thanks{Department of Economics, University of Essex.  e-mail: \texttt{m.freer@essex.ac.uk}} \and Hassan Nosratabadi\thanks{ECARES, Universit\'e Libre de Bruxelles. e-mail: \texttt{seyed.nosratabadi@ulb.be}}}

\begin{document}

\maketitle

\begin{abstract}
An observer wants to understand a decision-maker's welfare from her choice. She believes that decisions are made under limited attention. 
We argue that the standard model of limited attention cannot help the observer greatly. 
To address this issue, we study a family of models of choice under limited attention by imposing an \textit{attention floor} in the decision process. 
We construct an algorithm that recovers the revealed preference relation given an incomplete data set in these models. 
Next, we take these models to the experimental data.
We first show that assuming that subjects make at least one comparison before finalizing decisions (that is, an attention floor of 2) is almost costless in terms of describing the behavior when compared to the standard model of limited attention.
In terms of revealed preferences, on the other hand, the amended model does significantly better. We can not recover any preferences for 63\% of the subjects in the standard model, while the amended model reveals some preferences for all subjects. In total, the amended model allows us to recover one-third of the preferences that would be recovered under full attention. 
\end{abstract}

\section{Introduction}
An observer wants to retrieve reliable information about the preferences of a decision maker (henceforth, DM) from her choice. If DM is rational -- i.e., maximizes a complete and transitive preference relation -- then the task is straightforward. She concludes that the chosen alternative in a budget set is preferred to all the alternatives that are available but not chosen. However, this conclusion is only valid if DM pays attention to all of the alternatives in a problem, which is hardly satisfied in reality. This behooves the observer to drop the assumption of full attention in practice.

\cite{masatlioglu2012revealed} proposes a parsimonious model of choice under limited attention. In this model, DM is endowed with a rational preference relation, but only considers a subset of a given budget set. She chooses the best alternative among the ones she considers. The observer, generally speaking, may not know the exact manner in which DM's consideration sets are formed. However, if the behavioral pattern of interest is inattention, then a fundamental property should be that removing an alternative that DM does not pay attention to does not affect her choice. Consideration sets that satisfy this property are referred to as \emph{attention filters}, and the choice procedure described by a pair consisting of an attention filter and a rational preference relation is referred to as choice under limited attention, by \cite{masatlioglu2012revealed} (henceforth, the generic CLA).

How much do we learn about DM's preferences by observing her choice if we assume she follows the generic CLA model? Consider the following choice observations: $\C \{x,z\} = x$, $\C \{x,y,w\} = x$, and $\C \{y,w\} = y$. There is no violation of rational choice in these observations. Let us first see what we learn if DM is fully attentive. We would infer that $x$ maximizes DM's welfare since it is chosen over all other alternatives at some budget set. How does $x$ do if we instead assume that her attention is limited in the sense of the generic CLA? We would not retrieve \emph{any} information on how good $x$ is. This is because we could see these observations  in the following way: DM is disregarding everything except the chosen alternative, and thus her preference relation is completely unidentified.\footnote{It is easy to check that this interpretation is consistent with the definition of attention filters.} In other words, the \textit{welfare optimum} under full attention is \textit{no better than any other alternative} under the generic CLA model.

We should argue next that the example above is not just an abstract construct. In fact, there are two structural channels in the generic CLA model that restrict welfare analysis. The first channel is due to the fact that the \say{only paying attention to the choice alternative} interpretation is a valid one under the generic CLA model for any observation that does not violate rationality. At the same time, the experimental evidence normally finds a sizable proportion of subjects to be rational. Thus the first challenge is that no revealed preference analysis is possible for the sizable rational patterns in the data.\footnote{
This proportion, for example, is 75\% in \cite{manzini2010choosing}, 78\% in \cite{huber1982adding}, 60\% in \cite{doyle1999robustness}, and 33\% in \cite{inoue2018limited}.} The second channel is caused by the fact that, not only welfare revelation is conditioned on choice inconsistencies in the generic CLA model, but more specifically on a particular type of violation. \cite{masatlioglu2012revealed} characterizes these instances to be those where removing an alternative from a problem produces a choice reversal. To summarize, we can retrieve information about DM's welfare \textit{if, and only if}, we observe choice inconsistencies of this latter form.


Is there a way to avoid this collapse in welfare relevance without sacrificing the plausible assumption of limited attention? We answer this question in this paper by proposing what we perceive (and then confirm) to be a \emph{safe} adjustment to the generic CLA model. We argue that the \say{only paying attention to the choice alternative} interpretation is not a fair description of most economic applications since it, essentially, removes the act of \emph{comparison} from the decision process. In our view, DM faces trade-offs in most choice scenarios. Hence, we make the orthodox assumption that she, at least, makes one comparison in her choice procedure; that is, pays attention to at least two alternatives. 

How much welfare relevance does this assumption buy us? Let us go back to the example above: $\C \{x,z\} = x$, $\C \{x,y,w\} = x$, and $\C \{y,w\} = y$. In the first budget set, DM considers both alternatives and thus reveals that she likes $x$ better than $z$. Her choice from the second budget set tells us that $x$ is better than $y$ \textbf{or} that $x$ is better than $w$. However, we do not know exactly which preference holds. If $x$ is indeed better than $y$, then the fact that $y$ is chosen over $w$ in the last pairwise observation -- is better than $w$ -- would imply that $x$ is also better than $w$. This would mean that $x$ is the revealed best. The reality nonetheless, may be less promising. It could be that $x$ is better than $w$ in the second observation, in which case $x$ is guaranteed to be the second-best alternative. Thus, we can say that $x$ is either the first or the second-best alternative. That is a significant improvement over the generic CLA model which provided no help in finding the welfare optimum.

Following the latter example, in this paper we propose a hierarchy of limited attention models by imposing an \textit{attention floor} in the decision process; that is, requiring DM to pay attention to a minimum number of alternatives before finalizing her choice. We refer to the CLA model in which DM pays attention to at least $k$ alternatives as the $k^{th}$-order CLA model, which incorporates the generic CLA as the first-order. 
Note that, assuming an attention floor of $k$ is equivalent to assuming that DM is making at least $k - 1$ comparisons.

We first characterize these models based on complete data sets (all subsets of the universal set of alternatives are observed). While we believe complete data sets are hardly attained in reality, we think an axiomatic study of the models given complete data sets teaches us the fundamentals of the behavior of DM were she to follow our model. Our characterization, in essence, combines an axiom from rational choice (SARP$^k$) with another from choice under limited attention (WARP(LA$^{k}$). The former establishes that DM should act rationally if the size of the budget set she faces does not surpass her attention floor. The latter concerns her choice on the remainder of budget sets, where she may exhibit inconsistencies. This axiom requires these violations to only occur if her attention set is affected.

Having established the axioms that address the fundamentals of DM's behavior, we then turn our attention to the practical case where the data set in hand is incomplete (only some subsets of the universal set are observed). Taking the common approach from the revealed preference theory with incomplete data sets, which dates back to \cite{afriat1967construction} \citep[see][for an overview of the literature]{chambers2016revealed}, we provide a revealed preference test for our models. In particular, we characterize the $k^{th}$-order CLA model using the mixed-integer programming which is a computationally tractable algorithm.

We then tackle our original question which is welfare analysis under limited attention. A well-known phenomenon in models such as the one we study is multi-representation: one can find multiple (preference, attention filter) pairs that rationalize the data. The revealed preference relation in such cases is naturally defined as those preference rankings that appear in all such representations. We provide an algorithm to recover the revealed preference relation in our models. To see the kind of challenges that our algorithm needs to overcome, let us dig a bit more into the issue of multi-representation. There are two sources for this issue in our setup. The first source is rather straightforward. Since we work with incomplete data sets, the preferences of DM on alternatives she never chooses will not be identified. This is an standard issue that also appears under rational choice and mainly concerns ``out-of-sample'' predictions. The second factor is however more subtle and due to the behavioral aspects of our model. In the $k^{th}$-order CLA model, a significant amount of revealed preference information come in the form of \textbf{or}. Consider our previous example once again: $\C \{x,z\} = x$, $\C \{x,y,w\} = x$, and $\C \{y,w\} = y$ . Assume that the theory under which we would like to do revealed preference analysis is the $2^{nd}$-order CLA. Then we know that $x$ is better than $z$ (first observation) and that $y$ is better than $w$ (last observation). The second observation however is an \textit{uncertain} revealed preference information: $x$ is better than $y$ \textbf{or} $x$ is better than $w$. This last piece of information does not \textit{directly} enrich the revealed preference relation since it contains ``or''. However, when combined with other information, it does. We know that $x$ is better than $y$ or $x$ is better than $w$. If $x$ is better than $y$, then since $y$ is better than $w$ we would conclude that $x$ is better than $w$. That is, $x$ is revealed preferred to $w$ in any case. Our algorithm manages to find these indirectly revealed preferences that are obtained from combining certain and uncertain information.



We then proceed to apply our model to the experimental data from \cite{inoue2018limited} to see (i) how restrictive is imposing attention floors in the decision process (descriptive power), and (ii) how much welfare relevance one could actually gain from it. 
We find evidence that the assumption could significantly boost the welfare relevance with, interestingly, almost no cost in the descriptive power. An illustrative example is when we move from the first-order CLA model (generic) to the second-order (at least one comparison). For the first-order CLA no preferences can be recovered for 63\% of subjects, and on average we can only reveal one comparison per person. For the second-order CLA, at least 5 (out of 45 possible) comparisons for all the subjects are recovered with the average number of comparisons per subject being 9 (out of 45 possible). 
All in all, setting the attention floor to 2, we can recover one-third of all the comparisons that are revealed under full attention (best case scenario for recovering preferences). This boost in welfare relevance is almost costless in terms of the descriptive power: less than one percent of the subjects who pass the test for the first-order CLA, fail the test for the second-order CLA.


We believe our paper is the first to approach the phenomenon of limited attention from the aspect of welfare relevance (that is \textit{how much} welfare analysis one can do), especially in the case of incomplete data sets. Indeed, \cite{masatlioglu2012revealed} itself mentions that if one assumed that the attention sets contain at least two alternatives, then one would fully retrieve the preferences of DM from observations over doubletons. This conclusion only works if the observer has access to a data set which is too large even for an experimental setting with a large number of alternatives.\footnote{
    With 10 alternatives, one would need 45 observations only on doubletons. A complete data set will consist of 1013 problems.
}

A number of variations of choice under limited attention have been studied in the literature. Among these, there are models that directly impose size restrictions on the attention spans. However, these models mirror our notion and consider \textit{CLA with a capacity limit} instead. Two examples are \cite{geng2021shortlisting} and \cite{geng2022limited} where an ``attention cap'' is imposed on the first stage of decision-making. Our concept, instead is that of ``attention floor''. Another related paper is \cite{barbera2019order} where the authors develop a model of order-$k$ rationality in which DM is assumed to be endowed with a preference relation, but chooses any alternative that is ranked in the top-$k$ based on her preference relation in a given choice problem.

Another important feature of our paper is that we consider incomplete data sets which have recently received more attention in the literature. \cite{de2021bounded} characterizes the generic CLA model on incomplete data sets and also illustrates the challenges that the multi-representation problem brings in this model. Another important results is that of \cite{inoue2018limited} that characterizes several variations of the limited attention model proposed in the literature though on incomplete data sets.\footnote{Let us also note that \cite{dean2017limited} and \cite{maniquet2022welfare} study other behavioral models on incomplete data sets.}\cite{inoue2018limited} also does an empirical investigation into the descriptive power of these models. Let us also mention another paper that characterizes choice under limited consideration with incomplete data set, which is \cite{demuynck2018revealed}. However, these authors consider a setting with prices which is a version that avoids the multi-representation problem.

Finally, our paper also contributes to the growing literature on behavioral welfare economics and, particularly, the seminal work of \cite{bernheim2009beyond}. These authors adopt a model-free approach and focus on constructing a conservative notion of welfare based on observed choice. However, this approach would generate a coarse measure of welfare \citep[see][]{rubinstein2012eliciting,manzini2014welfare}. Moreover, \cite{masatlioglu2012revealed} points out a potential welfare conflict between \cite{bernheim2009beyond}'s model-free approach and a model-based one such as \cite{masatlioglu2012revealed}. Of course, since our model covers the generic CLA model as a special case, this critique carries over to our paper.\footnote{
For other examples of behavioral welfare economics without a model see \cite{chambers2012choice}, \cite{apesteguia2015measure}, and \cite{nishimura2018transitive}. See \cite{manzini2014welfare} for an argument for model-based approaches in behavioral welfare economics.} Let us also note that the empirical performance of some of these welfare notions have been investigated by \cite{bouacida2021predictive}.

The remainder of the paper is organized as follows.
Section \ref{sec:def} defines the $k^{th}$-order CLA model.
Section 3 provides the theoretical results: characterization on complete and incomplete data sets as well as the algorithm to recover the revealed preference relation.
Section 4 presents the experimental application.
Section 5 concludes.
All proofs are provided in the Appendix.

\section{Definitions}\label{sec:def}
Let $X$ be a finite universal set of alternatives. Let $\gB\subseteq 2^X$ be a collection of non-empty budget sets, where $2^X$ is the set of all subsets of $X$. $\C:\gB\rightarrow X$ is a \emph{choice function} if $\C(B)\in B$. Thus a \emph{data set} can be described as the pair $(\gB,\C)$. A binary relation $\succ\; \subseteq X\times X$ is a strict preference relation if it is complete (all elements are comparable), transitive ($x\succ x'$ and $x\succ x'$ implies $x\succ x''$ for any $x,x',x''\in X$), and asymmetric ($x\succ x'$ implies $x'\not\succ x$ for any $x,x'\in X$). 

Following the prominent paper of \cite{masatlioglu2012revealed}, we define \emph{attention filter} as $\Gamma: 2^X\rightarrow 2^X$ such that $\Gamma(B)\subseteq B$ and $\Gamma(B) = \Gamma(B\setminus\{x\})$ for every $B\in 2^X$ and $x\notin \Gamma(B)$. According to this model when DM suffers from limited attention, she chooses the best alternative according to her preference relation, $\succ$, from the subset that she pays attention to according to her attention filter, $\Gamma$. 

In the latter model, therefore, DM may for example only pay attention to the alternative that she is choosing ($\Gamma(B) = \C(B)$). 
As mentioned in the introduction we aim to make an adjustment to this theory so that it \emph{necessitates comparisons before making a final decision}. 
This would be equivalent to imposing an ``attention floor'', that is, a minimum threshold for the number of alternatives DM pays attention to before making her choice. Let us elaborate on this. Note that if this floor is set to 1, then DM would not need to make any comparisons. If, on the other hand, it is set to 2, then she makes at least one comparison. How about 3? Since she chooses the best from the alternatives she pays attention to, she would only need to make two comparisons here. For example, assume that $a,b$, and $c$ are the alternatives DM pays attention to. She can start by comparing $a$ and $b$ and then compare the ``winner'' with $c$. Transitivity would guarantee that she finds the best among the three alternatives this way. This line of reasoning means that the assumption that DM's attention floor is $k$ is equivalent to the assumption that DM makes at least $k-1$ comparisons. On this note, let us formally present the model.

\begin{defn}\label{def:1} A data set $(\gB,\C)$ is \textbf{rationalizable with $k^{th}$-order choice under limited attention} ($k^{th}$-order CLA) if there exist
\begin{itemize}
    \item [--] a strict preference relation $\succ$, and
    \item [--] an attention filter $\Gamma$ with the property that $|\Gamma(B)| \geq min\{|B|,k\}$ for all $B$,
\end{itemize} and such that
$$
\C(B) = \lmax{\succ}{\Gamma(B)}
$$
for any budget set $B\in \gB$.\footnote{
    Note that in this definition there is a uniform floor of the attention filter's size across menus. All the results presented can be immediately generalized to the case with heterogeneous caps. Thus, it is also possible to consider the theories with relative (or mixed) restrictions on the cardinalities of the attention filters. For instance, one can assume that DM makes at least one comparison and considers at least 10\% of the alternatives present in any budget set.}
\end{defn}

\noindent
According to Definition \ref{def:1}, the model of \cite{masatlioglu2012revealed} is the \emph{first-order CLA} and lies on one extreme in this hierarchical division. As the order of CLA representation increases, DM makes more comparisons in her decision process. In the other extreme, we have rational choice (fully attentive) when $k\ge \max_{B\in \gB} |B|$.

\section{Theoretical Results}\label{sec:chr-c}
We start by characterizing the class of models proposed in Definition \ref{def:1} under the assumption that the observer has access to a complete data set. That is the data set that includes all the possible choice problems DM could face from a universal set of alternatives. Having access to complete data sets, of course, is seldom the case. However one can learn important insights into the underlying principles of DM's behavior from axioms that operate on complete data sets. 
In particular, $k^{th}$-order CLA can be perceived as a ``convex combinatio'' between the CLA model and rational choice model. This fact is well illustrated by the axiomatization we provide here, where one axiom comes from restricting revealed preference axiom for rational choice and other comes from adjusting the revealed preference axiom under limited attention (see \cite{masatlioglu2012revealed}).

After dealing with this scenario, we move to the practical case, i.e. incomplete data sets, and show how an observer can see whether or not a given behavior is an instance of $k^{th}$-order CLA model. Note that, here we are focused on providing an empirical test rather than understating the underlying fundamentals of decision making. In the last part of this section, we define the revealed preference relation under limited attention and provide an algorithm that recovers this relation in a computationally efficient manner.

\subsection{Complete Data Set}
There are two underlying fundamentals in our DM's behavior. The first one accounts for DM's behavior on those sets where she is fully attentive; that is, where the size of the budget set she faces does not exceed her attention floor. On such budgets, she should act completely rational. Hence, her behavior should satisfy the \emph{strong axiom of revealed preference}; i.e., should not exhibit any \emph{direct} and \emph{indirect} cycles in revealed preferences. To implement this axiom in our setting, we use the following terminology. We say $x_1$ is \emph{$k^{th}$-order directly preferred to} $x_2$ if there exists a budget set $B\in \gB: |B|\leq k$ containing $x_1,x_2$ where $x_1$ is chosen. We then say, $x_1$ is \emph{k$^{th}$-order indirectly preferred to $x_n$} if there exist $\{x_2,x_3\ldots,x_{n-1}\}$ for some $n\ge 2$, such that $x_i$ is $k^{th}$-order directly preferred to $x_{i+1}$. Then we have the following axiom.

\begin{ax}\textbf{$k^{th}$-Order SARP (SARP$^k$)} If $x$ is $k^{th}$-order indirectly preferred to $y$, then $y$ is not $k^{th}$-order directly preferred to $x$.
\end{ax}

\noindent
It is worthy to make the following observation on SARP$^k$. If DM's attention floor is $2$, then SARP$^k$ only applies to doubletons where it imposes an acyclicity condition. Thus, SARP$^k$ reduces to the well-known \emph{no-binary-cycles} axiom in the special case of $k=2$. Of course, as her attention floor expands, SARP$^k$ demands more than no-binary-cycles. In the extreme case where DM is fully attentive in all problems -- i.e., rational -- SARP$^k$ would be equivalent to the classical SARP.

The second fundamental, naturally, pertains to DM's choice on those problems where she faces more alternatives than she pays attention to. In these cases, DM does not have to be fully rational. But still, some level of consistency is expected. To elaborate, note that rational choice admits to the following \emph{weak axiom} type condition: the choice in a problem is insensitive to removing any alternative other than the choice alternative itself. While in our model DM may generally violate this principle, she should still remain partially consistent with it. To be precise, DM's choice should be insensitive to the removal of an alternative that she \emph{does not pay attention to}, since such a change does not influence her attention set. But how can we identify if an alternative is not paid attention to in a given budget set? This can be ensured by concentrating on those problems where she is fully attentive. That is, budget sets that have at most $k$ alternatives. If she chooses an alternative $y$ over $x$ in a budget set where she is fully attentive -- which in turn means that she prefers $y$ to $x$ -- then she can only choose $x$ over $y$ in another problem if she is not paying attention to $y$. Hence, removing $y$ in this latter problem should not induce a choice reversal. This idea is captured in the next axiom.

\begin{ax} \textbf{WARP for $k^{th}$-Order Limited Attention (WARP(LA$^k$)})
For $x,y \in S\cap T$ with $|T| \leq k$, if $x = \C(S)$ and $y = \C(T)$, then $x = \C(S\setminus\{y\})$.
\end{ax}

\noindent
Recall that, when DM is fully attentive in all budgets, then SARP$^k$ is equivalent to classic SARP which, in turn, is sufficient to characterize the choice. Thus, logically speaking, WARP(LA$^k$) must be ``void'' in this case. This is true. If DM's behavior is rational then the main trigger for WARP(LA$^k$) -- that is, choice reversals going form one menu to the other -- is never observed. \

The two axioms above are necessary and sufficient to characterize $k^{th}$-order limited attention model, as stated in the next theorem.

\begin{thm}\label{thm:fulldata}
A complete data set is rationalizable with $k^{th}$-order CLA ($k\geq 2$) if and only if it satisfies SARP$^k$, and WARP(LA$^k$).
\end{thm}

\subsection{Incomplete Data Sets}\label{sec:chr-inc}
In this section, we turn our attention to incomplete data sets. The seminal work of \cite{afriat1967construction} on revealed preference theory utilizes linear programming to implement revealed preference tests on any given data set (complete or incomplete). The major advantage of using linear programming is computational feasibility. However, the limited attention theory is not compatible with the standard linear programming techniques. The major difference is that in the standard revealed preference one can draw inferences with ``certainty'' from choice.
That is, 
$$
\text{if}\; x = \C(B_i) \text{ and } y \in B_i, \; \text{then}\; x \; \text{is better than}\; y.
$$
However, revealed preference inference may be \textit{uncertain} under limited attention model. There are two reasons for this uncertainty. The first is illustrated by \cite{de2021bounded}. These authors argue that revealed preferences of the following form are critical in characterizing the generic CLA model on an incomplete data set. 

\begin{align*}
\C(B_j) \in B_i, \ \C(B_i) \in B_j, & \text{ and } \C(B_i)\ne \C(B_j) \text{ then at least one of the following holds:} \\ 
& \text{(i) there is } y\in B_i\setminus B_j \text{ such that } y\in \Gamma(B_i), \text{\bf or } \\
& \text{(ii) there is } z\in B_j\setminus B_i \text{ such that } z\in \Gamma(B_j). 
\end{align*}

\noindent
Note that this revealed preference information is uncertain because it is expressed using \textbf{or}. The expression is indeed a direct result from the property that defines an attention filter. If we remove an alternative that is not being paid attention to, then the attention filter should not change. Thus, if we observe the choice inconsistency between $B_i$ and $B_j$, then it has to be the case that while moving from one budget to another (either from $B_i$ to $B_j$ or $B_j$ to $B_i$) we should remove something that is being paid attention to. (otherwise the choice should not change in this process which can not be the case since we observe a reversal)
This means that, revealed preferences under limited attention may come in \textbf{or} form. In addition to this observation from \cite{de2021bounded}, we should note that $k^{th}$-order CLA model directly adds another channel for the ``or-logic'' in revealed preference. For example, if we observe $\C\{x,y,z\} = x$ and we believe that DM's attention floor is $2$, then we will have revealed preference information in the form: $x \succ y$ or $x \succ z$. 

To account for the presence of the or-logic in the formulation of revealed preferences, we employ mixed integer programming (MIP). More precisely, we amend the linear program with some integer (binary) variables to deal with the or-logic. Let us introduce the notation necessary to introduce the MIP. We enumerate budgets in $\gB$ by $B_i$, and let $u_i>0$ be the utility value of the chosen alternative in $B_i$ ($\C(B_i)$). The second set of variables are the integers $\delta_{ij}\in \lbrace 0,1\rbrace$ defined as $\delta_{ij}=1$ if $x_j\in \Gamma(B_i)$ and $\delta_{ij}=0$ otherwise. Hence, $\delta_{ij}=1$ if an alternative chosen in $B_j$ ($\C(B_j)$) is paid attention to in $B_i$. Finally, we introduce the big-M ($M\rightarrow \infty$) that is a technical construct to implement the or-logic in the program, and also denote by $b_i$ the cardinality of $B_i$. 
We can state the mixed integer program (MIP) as follows.
\begin{equation}
\label{eq:MIP}
    \begin{cases}
    u_i > u_j - (1-\delta_{ij})M & \text{ if } x_i = \C(B_i), x_j \in B_i, \text{ and } x_i \ne x_j
    \\
    \delta_{ij} = 0 & \text{ if }  x_j \notin B_i
    \\
    \delta_{ij} = 1 & \text{ if }  x_j = \C(B_i)
    \\
    \sum\limits_{x_k \in B_i\setminus B_j}\delta_{ik} + \sum\limits_{x_k \in B_j\setminus B_i}\delta_{jk} \ge 1 & \text{ if } \C(B_i)\ne \C(B_j), \ \C(B_i),\C(B_j)\in B_i\cap B_j
    \\
    \sum\limits_{j\ne i} \delta_{ij} \ge \min\lbrace k,b_i\rbrace & 
    \end{cases}
\end{equation}

System \eqref{eq:MIP} presents a straightforward test for the class of models we constructed. The first line implements the idea that the chosen alternative is better than any other alternative in the consideration set. 
The second line states that DM only considers alternatives if they are in $\Gamma(B_i)\subseteq B_i$. The third line guarantees that DM is paying attention to the chosen alternative. The forth line controls for the important observation on limited attention emphasized by \cite{de2021bounded}. Note that the first four lines of the system \eqref{eq:MIP} together specify a MIP formulation for the \cite{de2021bounded} test of the generic CLA model. Thus, we only need to include the part corresponding to the attention floor. This is done by the last line.


It is evident that the conditions in system \eqref{eq:MIP} are necessary for rationalizability of the data with $k^{th}$-order CLA model. The following theorem shows that they are also sufficient.

\begin{thm}
\label{prop:RationalizabilityIncompleteExperiment}
A data set $(\gB,\C)$ is rationalizable with $k^{th}$-order CLA if and only if there is a solution to equation \eqref{eq:MIP}.
\end{thm}

\subsection{The Revealed Preference Relation}
It is known and also straightforward to argue that one runs into multi-representation problem in our setup; that is, one can find multiple pairs of (preferences, attention filter) that rationalize an observed behavior. In presence of multi-representation, the revealed preference relation is naturally defined by the rankings that appear in all representations. To formalize this idea for the model we study, let $\cM(\gB,\C)= \{(\succ, \Gamma) \}$ be the set of pairs (preferences, attention filter) that rationalize the data set $(\gB,\C)$ according to $k^{th}$-order CLA. 
Then, the revealed preference relation is
$$
P^* = \bigcap_{(\succ,\Gamma)\in \cM(\gB,\C)} \succ;
$$
that is, the intersection of all possible preference relations that rationalize the data set. Note the conservative nature of this definition as it aims to report revealed preferences that are guaranteed to be true: $xP^* y$ if and only if $x \succ y$ for every $\succ$ that generates the data.




A brute-force way to find $P^*$ requires checking every $(x,y) \in X\times X$ to see whether $y \succ x$ hold for some $(\succ,\Gamma)\in \cM(\gB,\C)$. This algorithm becomes computationally intractable very quickly.
Thus, we need to restrict the potential number of comparisons. We use the following strategy to do so.
First, we consider a \textit{lower bound for $P^*$}: the transitive closure of all certain revealed preference information -- the ones that are not expressed in or-logic. 
Next, we find a ``minimal'' binary relation that rationalizes the data. 
This would serve as a \textit{upper bound for $P^*$} since any comparison $x P^* y$ has to appear in every rationalization. 
We know that $P^*$ lies between these two bounds. 
Thus, we are only left to check the comparisons that are in the upper bound but not in the lower, one by one, to see whether or not they are in $P^*$. 

\paragraph{Lower Bound.}
We start by constructing a lower bound for $P^*$ denoted by $\ubar P$. 
That is, if $x \ubar P y$, then $x P^* y$ ($P^* \subseteq \ubar P$). 
Given the definition of attention filters, we know that 
$$
\C(B) \neq \C(B\setminus \{ y\}) \text{ implies that } \C(B) \succ y
$$
for every $(\succ, \Gamma)\in \cM(\gB,\C)$.
Moreover, imposing an attention floor equal to $k$, we also retrieve certain revealed preference information whenever the cardinality of the set does not exceed $k$.
Finally, since $P^*$ is transitive, then every comparison that belongs to the transitive closure of $\ubar P$ also belongs to $P^*$.

\begin{defn}
\label{def:LowerBound}
Assume that the data set is rationalizable with $k^{th}$-order CLA. We define $x \ubar P y$ if and only if
\begin{itemize}
    \item [--] there exists $B$ such that $B, B\setminus \{y\} \in \gB$, and $x=\C(B) \ne \C(B\setminus \{y\})$, or
    \item [--] there exists $B\in \gB$ with $|B| \leq k$ such that $x=\C(B)$, $y\in B$, or
    \item [--] there exist $x=s_1, \ldots, s_n = y$ such that $s_j \ubar P s_{j+1}$ fore very $j\in \{1,\ldots, n-1\}$.
\end{itemize}
\end{defn}

The natural question next is whether there is some information about $P^*$ contained outside $\ubar P$. That is, can we elicit more information about $P^*$ beyond $\ubar{P}$ from the observed choice using \textit{uncertain} revealed preference information that are expressed in or-logic? Intuitively speaking, the answer is yes. For example, assume that we know with certainty that $x \succ y$ and $z \succ y$. Also, assume that we have the following uncertain revealed preferences: $t \succ x$ or $t \succ z$. Then we can combine the certain and uncertain revealed preferences to conclude, for certain, that $t \succ y$. Thus, the information contained in $\ubar P$ is not sufficient to recover $P^*$ fully.


\paragraph{Upper Bound}
By an upper bound for $P^*$ we mean a relation $\bar P$ such that if $x P^* y$, then $x \bar P y$ ($P^* \subseteq \bar P$). 
Note that any binary relation that rationalizes the data could work as an upper bound. Nevertheless, in order to minimize the number of checks between these bounds, it is important to construct a ``minimal'' upper bound. 
We do this by finding the minimal binary relation that is needed to rationalize the data. We proceed by minimizing the lower contour set for every alternative that has been chosen from at least one budget set. To do this, let us modify the system \eqref{eq:MIP} to construct a new MIP. 

Let $b_i = |B_i|$ be the cardinality of the set $B_i\in \mathcal B$. 
Let $s_i \in \{0,1\}$ be the indicator of whether $x_i$ is in the (weak) lower contour set of $x$. 
Let $\delta_{ij} \in \{0,1\}$ be the indicator that is equal to one if $x_j \in \Gamma(B_i)$ for $B_i\in \gB$ and zero otherwise. 
We have the following system.

\begin{equation}
\label{eq:LowerContourSetMIP}
    \begin{cases}
    \sum_i s_i \rightarrow \min & \\
    s_i = 1         & \text{ if } x_i = x 
    \\
    s_j + (1-\delta_{ij})M\ge s_i & \text{ if } x_i = \C(B_i), \ x_j \in B_i, \text{ and } x_i \ne x_j
    \\
    \delta_{ij} = 0 & \text{ if }  x_j \notin B_i
    \\
    \delta_{ij} = 1 & \text{ if }  x_j \notin \C(B_i)
    \\
    \sum\limits_{x_k \in B_i\setminus B_j}\delta_{ik} + \sum\limits_{x_k \in B_j\setminus B_i}\delta_{jk} \ge 1 & \text{ if } \C(B_i)\ne \C(B_j), \ \C(B_i),\C(B_j)\in B_i\cap B_j
    \\
    \sum\limits_{x_j\in B_i} \delta_{ij} \ge \min \lbrace b_i,k \rbrace &
    \end{cases}
\end{equation}

\noindent
Let $\hat s$ be the solution to system \eqref{eq:LowerContourSetMIP}. Then, we can define the minimal lower contour set of $x$ as
$$
\hat L(x) = \{x_i\in X: s_i = 1 \text{ and } x_i \ne x \}.
$$

\begin{defn}
\label{def:UpperBound}
Let data set be rationalizable with $k^{th}$-order CLA, then $x \bar P y$ if and only if $y \in \hat L(x)$. 
\end{defn}

\paragraph{Algorithm.}
Having defined lower and upper bounds for $P^*$ we now can improve the algorithm for recovering $P^*$. 
We start by making an important remark. We can initialize system \eqref{eq:MIP} with some known preferences. In particular, we know that if the data set is initially rationalizable, then it would still be if we initialize \eqref{eq:MIP} given $\ubar P$. Thus, the idea is to try adding reversed comparisons from $\bar P \setminus \ubar P$ to $\ubar P$ one by one and see whether the data set is rationalizable with this enriched basic preference relation (that is, $\bar P\cup \{(x,y)\}$) 
If the comparisons can be reversed in a non-contradictory manner, then it can not belong to $P^*$. Otherwise, it belongs. Once we tried all such comparisons the algorithm stops.

To formalize the algorithm let us use some auxiliary notation.
Let 
$$
\cM(\gB,\C,P) = \{ (\succ,\Gamma)\in \cM(\gB,\C) \text{ and } P\subseteq \; \succ \}
$$
be the set of pairs (attention filter, preferences) that rationalize the data such that the preference relation in the pair contains $P$. Note that in the language of $\cM$ the data set is rationalizable if and only if $\cM(\gB,\C) \ne \emptyset$. Similarly rationalizability given an initial revealed preference information $P$ is equivalent to $\cM(\gB,\C,P) \ne \emptyset$.

\begin{algorithm}[htb]
\caption{Recovering $P^*$}
\label{algo:EstimatingPreferences}
\begin{algorithmic}[1]
\Require $\cM(\gB,\C)\ne \emptyset$
\State $\triangle \leftarrow (\bar P \setminus \ubar P)$ 
\State $\hat P  \leftarrow \ubar P$
\While{$\triangle P \ne \emptyset$}
    \State $(x,y) \in \triangle \succ$
    \If{$\cM(\gB,\C, \hat P \cup \{(y,x)\})=\emptyset$}
        \State $\hat P \leftarrow \hat P \cup \{(x,y)\}$
    \EndIf
    \State $\triangle P \leftarrow \triangle P \setminus \{(x,y)\}$
\EndWhile
\end{algorithmic}
\end{algorithm}


\begin{thm}
\label{thm:CorrectAlgorithm}
Let $\hat P$ be the output of algorithm \ref{algo:EstimatingPreferences}.
If a data set is rationalizable, then $\hat P = P^*$.
\end{thm}

\section{Experimental Application}
We use the data from the experiment conducted by \cite{inoue2018limited}.\footnote{
    Let us note that this data set has the lowest pass rates (for rational behavior) among the data sets we referred to in the introduction (33\% of subjects are rational). Recall that if an observation is consistent with rational behavior, then the generic CLA has no welfare implications. Hence, in terms of welfare relevance, the generic CLA's performance would have been only worse if we used other data sets.
}
Let us present the essential details of the experiment. Subjects are asked to choose between bundles of intertemporal installments. Each bundle consists of three installments: in 1 month, in 3 months, and in 5 months. Each renumeration bundle consists of 2400 Japanese yen, split into three installments. The universal set consists of 10 alternatives presented in Table \ref{tab:Alternatives}. Each subject was presented with 20 different budgets to choose from. Budgets were generated to contain 2 to 8 alternatives. A total of 113 subjects (students of Waseda University, Japan) participated in the experiment that was run in 4 sessions.

\begin{table}[ht]
    \centering
    \begin{tabular}{l|cccccccccc}
         &  $x_1$ & $x_2$ & $x_3$ & $x_4$ & $x_5$ & $x_6$ & $x_7$ & $x_8$ & $x_9$ & $x_{10}$ \\ \hline
    in 1 month & 450 &  800 & 1150  & 450 &  450 & 800 & 850 & 1200 & 1550 & 500 \\
    in 3 months & 800 & 800 & 800 & 450 & 1500 & 1150 & 0 & 0 & 0 & 0 \\
    in 5 months & 1150 & 800 & 450 & 1500 & 450 & 450 & 1550 & 1200 & 850 & 1900 \\ \hline
    \end{tabular}
    \caption{Universal set of alternatives. Payments are in Japanese yen.}
    \label{tab:Alternatives}
\end{table}

\noindent
We present the analysis of the data in three parts. We start by testing different orders of the CLA model. This helps us to understand to what extent an order of the CLA theory is viable in explaining the observed behavior (descriptive power). Next, we evaluate the welfare relevance of the theory. Finally, we analyze to which extent we can use $\ubar P$ as a good approximation of $P^*$.


Given that the maximal size of the budget set is 8, we are going to present results for all 8 orders of the theory, which range from $\vert \Gamma \vert \ge 1$ to $\vert \Gamma \vert \ge 8$. The first version ($\vert \Gamma \vert \ge 1$) is equivalent to the generic CLA model of \cite{masatlioglu2012revealed}, and the last ($\vert \Gamma \vert \ge 8$) is equivalent to rational choice since all alternatives are considered. 


\subsection{Descriptive Power}

\begin{table}[htb]
\centering
\begin{tabular}{l|c|cc|cc}
\toprule
                             & Pass Rate                                                & \multicolumn{2}{c|}{Power}                                                                                                                                   & \multicolumn{2}{c}{PSI} \\
                             &                                                          & Bronars                                                                      & Bootstrap                                                                     & Bronars   & Bootstrap   \\ \hline
$\vert \Gamma \vert \ge 1$   & 1                                                        & \begin{tabular}[c]{@{}c@{}}.96\\ (.95; .97)\end{tabular}                     & \begin{tabular}[c]{@{}c@{}}.88\\ (.87; .91)\end{tabular}                      & .04       & .12         \\
$\vert \Gamma \vert \ge 2$ \ & \begin{tabular}[c]{@{}c@{}}.99\\ (.95; .99)\end{tabular} & \begin{tabular}[c]{@{}c@{}}.46\\ (.43; .49)\end{tabular}                     & \begin{tabular}[c]{@{}c@{}}.35\\ (.32; .38)\end{tabular}                      & .53       & .64         \\
$\vert \Gamma \vert \ge 3$   & \begin{tabular}[c]{@{}c@{}}.88\\ (.80; .93)\end{tabular} & \begin{tabular}[c]{@{}c@{}}.24\\ (.22; .27)\end{tabular}                     & \begin{tabular}[c]{@{}c@{}}.17\\ (.04; .07)\end{tabular}                      & .64       & .71         \\
$\vert \Gamma \vert \ge 4$   & \begin{tabular}[c]{@{}c@{}}.69\\ (.59; .77)\end{tabular} & \multicolumn{1}{l}{\begin{tabular}[c]{@{}c@{}}.01\\ (.01; .02)\end{tabular}} & \multicolumn{1}{l|}{\begin{tabular}[c]{@{}c@{}}.02\\ (.01; .03)\end{tabular}} & .68       & .67         \\
$\vert \Gamma \vert \ge 5$   & \begin{tabular}[c]{@{}c@{}}.61\\ (.51; .70)\end{tabular} & \multicolumn{1}{l}{\begin{tabular}[c]{@{}c@{}}.00\\ (.00; .01)\end{tabular}} & \multicolumn{1}{l|}{\begin{tabular}[c]{@{}c@{}}.00\\ (.00; .01)\end{tabular}} & .61       & .61         \\
$\vert \Gamma \vert \ge 6$   & \begin{tabular}[c]{@{}c@{}}.50\\ (.41; .60)\end{tabular} & 0                                                                            & 0                                                                             & .50       & .50         \\
$\vert \Gamma \vert \ge 7$   & \begin{tabular}[c]{@{}c@{}}.40\\ (.31; .49)\end{tabular} & 0                                                                            & 0                                                                             & .40       & .40         \\
$\vert \Gamma \vert \ge 8$   & \begin{tabular}[c]{@{}c@{}}.34\\ (.25; .43)\end{tabular} & 0                                                                            & 0                                                                             & .34       & .34    \\ \bottomrule    
\end{tabular}
\caption{Descriptive power. Pass rates and power computations contain the 95\% confidence intervals in the parenthesis. All confidence intervals are computed using the Clopper-Pearson procedure. Confidence intervals for values of 0 and 1 are omitted as those are not defined.}
\label{tab:PassRates}
\end{table}

\noindent
Table \ref{tab:PassRates} presents the descriptive power of different orders of CLA. Each row corresponds to an attention floor. The second column presents the pass rates; that is, the share of the subjects who pass the test for each order of the theory. We see that all subjects are consistent with the first-order CLA model (\cite{masatlioglu2012revealed}) and, 99\% are with the second-order CLA model (one comparison). 
Numbers start dropping rapidly as the attention span grows, going all the way down to 33\% for the case of full attention ($|\Gamma|\ge 8$).

Given that different tests impose different restrictions, we need to identify the \say{power} of the test. In words, we need to establish the probability that a random data set is consistent with $k^{th}$-order CLA model. We consider two standard power tests in the literature: Bronars’ power \citep[see][]{bronars1987power} and bootstrap power \citep[see][]{andreoni2013power}.
The difference between these tests is the distributional assumption on the random choice. 
Bronars' power creates random choice according to a uniform distribution among all available options. Bootstrap power uses the distribution of observed choice to generate random data set 
To compute each of these power indices, we generate 1000 random subjects and estimate their pass rate.
We observe that 96\% of random Bronars subjects and 88\% of random bootstrap ones are consistent with the first-order CLA model. Note that the change moving from the first-order to the second is drastic as the number of random subjects consistent with the model drops from 88-96\% to 35-46\% with further dropping to 17-24\% for $|\Gamma|\ge 3$. Starting with $\vert \Gamma \vert \ge 6$ none of the random subjects are passing the tests, showing that at this point our test becomes powerful enough to exclude \textit{false positives}. However, we do see that the pass rates among real subjects are also going quite low at this point.

Columns 5 and 6 present the predictive success indices \citep[PSI, see ][]{selten1991properties,beatty2011demanding}. PSI is computed by taking the difference between the pass rate and the power index. Therefore, it corrects the pass rate for the probability of false positive. PSI takes values between -1 and 1. If it is equal to -1, then all real subjects fail the test while all random subjects pass; and if it is equal to 1m then all real subjects pass the test while all random subjects fail. We present two versions of PSI for Bronars' and bootstrap power corrections.
Considering PSI performance, we see that the $3^{rd}$-order CLA ($|\Gamma|\ge 3$) and the $4^{th}$-order CLA ($|\Gamma|\ge 4$) models deliver the highest rates among all other models for bootstrap and Bronars' correspondingly. The second-order CLA model that only requires a single comparison performs comparably well to $|\Gamma|\ge 3$ and $|\Gamma|\ge 4$ delivering the PSIs of .53 using Bronars' and .64 for bootstrap controls. The PSIs for the first-order CLA model (.04 for Bronars and .12 for bootstrap) and fully attentive behavior (.34) are significantly lower than those for the $|\Gamma|\ge 4$, $|\Gamma|\ge 3$, and $|\Gamma|\ge 2$ theories.

\subsection{Welfare Relevance}
Next, we proceed with evaluating the welfare relevance of $k^{th}$-order CLA models. We consider two measures. The first measure is the relative density of the revealed preference relation. That is, the number of comparisons that are revealed under $k^{th}$-order CLA ($P^*$) as a share of the total number of comparisons a complete strict preference relation over $X$ would contain. For example, when the total number of alternatives is 4, then a total of 6 comparisons is contained in any complete strict preference relation. If $P^*$ reveals two, then our first measure reports: $2/6 = 33\%$. While being a straightforward measure of welfare relevance, the relative density nonetheless may be perceived as ``raw``, as it may not help the observer if she is concerned with finding, say, ``the best alternative''. We propose a second measure to address this. The second measure reports the maximal size of the lower contour sets across all alternatives; that is, it finds the alternative that dominates the largest number of other alternatives, and then reports to us that exact number. This measure therefore helps the observer to approximately identify the best alternative.


\begin{table}[ht]
\centering
\begin{tabular}{c|c|ccccc}
\toprule
& \multirow{2}{*}{\begin{tabular}[c]{@{}c@{}}Mean \\ (Standard Deviation)\end{tabular}} 
& \multicolumn{5}{c}{Quartiles of the Distribution} \\      
&                         & min      & 25\%      & 50\%     & 75\%     & max     \\   \midrule        
$\vert \Gamma \vert \ge 1$ & \begin{tabular}[c]{@{}c@{}}0.01\\ (0.02)\end{tabular}& 0   & 0   & 0    & 0.02     & 0.13       \\    
$\vert \Gamma \vert \ge 2$ & \begin{tabular}[c]{@{}c@{}}0.20\\ (0.07)\end{tabular}& 0.13  & 0.16 & 0.18  & 0.23  & 0.53   \\   
$\vert \Gamma \vert \ge 3$ & \begin{tabular}[c]{@{}c@{}}0.35\\ (0.08)\end{tabular}& 0.20 & 0.29 & 0.33  & 0.40 & 0.58       \\ 
$\vert \Gamma \vert \ge 4$ & \begin{tabular}[c]{@{}c@{}}0.44\\ (0.06)\end{tabular} & 0.29   & 0.40 & 0.42   & 0.47 & 0.62    \\     
$\vert \Gamma \vert \ge 5$ & \begin{tabular}[c]{@{}c@{}}0.53\\ (0.08)\end{tabular}& 0.42 & 0.47 & 0.49 & 0.62 & 0.76       \\  
$\vert \Gamma \vert \ge 6$ & \begin{tabular}[c]{@{}c@{}}0.61\\ (0.08)\end{tabular} & 0.42    & 0.53    & 0.60  & 0.69   & 0.78   \\    
$\vert \Gamma \vert \ge 7$ & \begin{tabular}[c]{@{}c@{}}0.62\\ (0.10)\end{tabular} & 0.42   & 0.53  & 0.60 & 0.71 & 0.80   \\        
$\vert \Gamma \vert \ge 8$ & \begin{tabular}[c]{@{}c@{}}0.63 \\ (0.09)\end{tabular} & 0.42   & 0.56    & 0.63 & 0.71   & 0.80         \\ \bottomrule               
\end{tabular}%
\caption{Relative density of the revealed preference relation.
Number of observed comparisons divided by the theoretical number of comparisons in the complete preference relations (45).
Numbers are only computed for the subjects consistent with the given theory.
}
\label{tab:RevaledPreferenceDensity}
\end{table}

Let us first look at the numbers with respect to our first measure presented in Table \ref{tab:RevaledPreferenceDensity}. Every row pertains to an attention floor. The second column presents the mean density (across subjects) of the revealed preference relation with standard deviation in parentheses. The columns after (in the right panel) present the quartiles of the distribution.
We start by making observations about the extremes.
For the generic CLA ($|\Gamma|\ge 1$) we see that there is no revealed preferences for more than 50\% of the sample (the exact number is 63\%) and no more than one comparison revealed for at least 75\% of the sample. In the other extreme, we have rational choice ($|\Gamma|\ge 8$) where we can, on average, retrieve 63\% of possible comparisons, with the maximum at 80\%. Thus, even under rational choice the complete preference relation is not recovered in this experiment. Considering the intermediate cases, if DM makes at least one comparison, then the theory is already powerful enough to recover 20\% of potential comparison on average. That accounts for about one- third of the best case scenario which is rational choice. Also, observe that the minimal density under $|\Gamma|\ge 2$ is equally large as the maximal density under the generic CLA. Considering the orders that perform best according to PSI -- $|\Gamma|\ge 3$ and $|\Gamma|\ge 4$ -- 35\% and 44\% of all comparisons on average are revealed which correspond to half and two-thirds of the densities under rational choice. Overall, we see that imposing attention floors allows us to reconstruct a significant chunk of underlying preferences. Once the minimal restriction goes up to five, we see that the difference between the density under $k^{th}$-order CLA and rational choice becomes rather negligible. To summarize the results in Tables \ref{tab:PassRates} and \ref{tab:RevaledPreferenceDensity}, we observe a trade-off between the descriptive power and welfare content as a function of DM's attention floor. This trade-off appears to be greatly favorable in the first steps (from the first-order up to the forth-order) as significant boosts in welfare content are achieved with not much loss in descriptive power. The trade-off becomes less appealing as we move to higher orders.

\begin{table}[htb]
\centering
\begin{tabular}{c|c|ccccc}
\toprule
                           & \multirow{2}{*}{\begin{tabular}[c]{@{}c@{}}Mean\\ (Standard Deviation)\end{tabular}} & \multicolumn{5}{c}{Quartiles of the Distribution} \\
                           &                                                                                      & min      & 25\%      & 50\%     & 75\%     & max     \\ \hline
$\vert \Gamma \vert \ge 1$ & \begin{tabular}[c]{@{}c@{}}0.44\\ (0.67)\end{tabular}                                & 0        & 0        & 0       & 1       & 4       \\
$\vert \Gamma \vert \ge 2$ & \begin{tabular}[c]{@{}c@{}}4.06\\ (0.90)\end{tabular}                                & 2        & 4        & 4       & 4       & 7       \\
$\vert \Gamma \vert \ge 3$ & \begin{tabular}[c]{@{}c@{}}5.59\\ (0.87)\end{tabular}                                & 4        & 5        & 5       & 6       & 8       \\
$\vert \Gamma \vert \ge 4$ & \begin{tabular}[c]{@{}c@{}}7.03\\ (0.99)\end{tabular}                                & 5        & 6      & 7       & 8       & 8       \\
$\vert \Gamma \vert \ge 5$ & \begin{tabular}[c]{@{}c@{}}8.58\\ (0.55)\end{tabular}                                & 7        & 8        & 9       & 9       & 9       \\
$\vert \Gamma \vert \ge 6$ & \begin{tabular}[c]{@{}c@{}}8.67\\ (0.48)\end{tabular}                                & 8        & 8        & 9       & 9       & 9       \\
$\vert \Gamma \vert \ge 7$ & \begin{tabular}[c]{@{}c@{}}8.98\\ (0.15)\end{tabular}                                & 8        & 9        & 9       & 9       & 9       \\
$\vert \Gamma \vert \ge 8$ & \begin{tabular}[c]{@{}c@{}}9.00\\ (0.00)\end{tabular}                                   & 9        & 9        & 9       & 9       & 9       \\ \bottomrule
\end{tabular}
\caption{Maximal lower contour set. Hence, 0 means that the alternative with the largest lower contour set might be the worst, and 9 means that it is actually the best. Numbers are only computed for the subjects consistent with the given theory.}
\label{tab:Welfare}
\end{table}



Next we go to Table \ref{tab:Welfare} that presents results from the size of the maximal lower contour sets across the theories. Let us first note that the experiment is \say{rich enough} in this dimension; that is, if we assume that $\vert \Gamma\vert \ge 8$ (rational choice), then the alternative with the largest lower counter set would indeed be the best alternative. Hence, under the assumption of full attention we can guarantee the welfare optimum using the data set from the experiment. In particular, note that this is despite the fact that only 63\% of the total number of rankings is revealed under full attention. (Recall Table \ref{tab:RevaledPreferenceDensity}) This, therefore, provides evidence for our earlier argument on the appeal of our second measure whenever the observer is interested in identifying the best alternative.

On the other extreme, the generic CLA model ($\vert \Gamma\vert \ge 1$) does not help us to elicit any comparisons for 63\% of the sample. Following our discussion in the introduction, we should note that the first systemic channel that denies welfare relevance for the CLA model -- subjects who are rational -- only covers 33\% of this. The remaining 30\% of subjects are cases of the second channel: those whose choice violate the standard axioms, but not in a pattern that helps us to retrieve preference under the generic CLA model.

As reported in Table \ref{tab:Welfare}, for the first-order CLA model, on average, we cannot even guarantee that the alternative that has the largest lower counter set in $P^*$ is the second worst. Looking across subjects indeed, the best case for this order of the model is where this alternative is better than four others. On the other hand, with attention floor set to 2 ($\vert \Gamma\vert \ge 2$), we gain a significant boost in welfare relevance using this measure, guaranteeing that such an alternative is, on average, better than at least four others. In the best case scenario across subjects at this order this alternative is at least the third-best. Moreover, we see that, at this order of the model, the size of the maximal lower contour set for any given alternative is strictly better for \textit{all} subjects compared to the basic CLA model.
Further increasing $k$, we see that the size of the maximal lower contour set when the attention floor is set to 3 is strictly better than the size of the maximal lower contour set when the attention floor is 2 for 97 subjects out of 99 (98\%) who are consistent with the former floor level. 

It is also seen in Table \ref{tab:Welfare} that the marginal gain in the the size of the maximal lower contour set rapidly decreases as the attention floor grows. This reduction in the marginal gain is most visible after $\vert \Gamma\vert \ge 5$. Finally, in regards to theories with seemingly the most reasonable trade-offs between descriptive power and welfare relevance, ($\vert \Gamma\vert \ge 2$ and $\vert \Gamma\vert \ge 3$), it is observed that the alternative with the maximal lower contour set is, on average, better than 4 (for $\vert \Gamma\vert \ge 2$) and 5 (for $\vert \Gamma\vert \ge 3$) others.\footnote{
    Note that when the attention floor is set at a number above 1, any doubletone budget reveals preference. This may suggest that the boost in welfare relevance in this experiment after imposing attention floors is mainly due to inclusion of pairs in the experiment. Importantly, this is not the case and the results in this section are robust to removing the doubletones from the experiment. This instead suggests that the indirect revelation yielded by our algorithm via combining certain and uncertain revealed preference information play a major role in a model with an attention floor.}

\subsection{Approximation of the Revealed Preference Relation}
Even though the algorithm \ref{algo:EstimatingPreferences} is computationally efficient in our application, in general it might take non-trivial time 
to converge. Thus, for scenarios with larger universal sets of alternatives and larger experiments one might prefer to use one of the computationally efficient bounds we offered. Note that the lower bound ($\ubar P$) is the most efficient one, as it is a simple iterative algorithm that does not require solving MIPs. Here we perform an analysis on how good of approximations these bounds can be for the the revealed preference relation.

\begin{table}[htb]
\centering
\begin{tabular}{c|ccc|c}
\toprule
                & \multicolumn{3}{c|}{\begin{tabular}[c]{@{}c@{}} Maximal Lower Contour Set \\ \hline  Mean\\ (Standard Deviation)\end{tabular}}                                                                                             & {\begin{tabular}[c]{@{}c@{}}
                $|P^*\setminus \ubar P|$ \\ \hline  Mean\\ (Standard Deviation) \end{tabular}} \\
      & $\ubar P$   & $P^*$        & $\bar P$    &   \\ \midrule
$|\Gamma|\ge 1$ & \begin{tabular}[c]{@{}c@{}}0.44\\ (0.67)\end{tabular} & \begin{tabular}[c]{@{}c@{}}0.44\\ (0.67)\end{tabular} & \begin{tabular}[c]{@{}c@{}}0.64\\ (0.89)\end{tabular} 
& \begin{tabular}[c]{@{}c@{}}0.02\\ (0.28)\end{tabular} \\
$|\Gamma|\ge 2$ & \begin{tabular}[c]{@{}c@{}}4.06\\ (0.90)\end{tabular} & \begin{tabular}[c]{@{}c@{}}4.06\\ (0.90)\end{tabular} & \begin{tabular}[c]{@{}c@{}}4.34\\ (0.98)\end{tabular} 
& \begin{tabular}[c]{@{}c@{}}0.26\\ (0.71)\end{tabular} \\
$|\Gamma|\ge 3$ & \begin{tabular}[c]{@{}c@{}}5.57\\ (0.85)\end{tabular} & \begin{tabular}[c]{@{}c@{}}5.59\\ (0.87)\end{tabular} & \begin{tabular}[c]{@{}c@{}}5.90\\ (0.95)\end{tabular} 
& \begin{tabular}[c]{@{}c@{}}0.71\\(1.78)\end{tabular}    \\
$|\Gamma|\ge 4$ & \begin{tabular}[c]{@{}c@{}}7.01\\ (0.99)\end{tabular} & \begin{tabular}[c]{@{}c@{}}7.03\\ (0.99)\end{tabular} & \begin{tabular}[c]{@{}c@{}}7.26\\ (0.86)\end{tabular} 
& \begin{tabular}[c]{@{}c@{}}0.56\\ (1.48)\end{tabular}     \\
$|\Gamma|\ge 5$ & \begin{tabular}[c]{@{}c@{}}8.58\\ (0.55)\end{tabular} & \begin{tabular}[c]{@{}c@{}}8.58\\ (0.55)\end{tabular} & \begin{tabular}[c]{@{}c@{}}8.58\\ (0.55)\end{tabular} 
& \begin{tabular}[c]{@{}c@{}}0.38\\ (1.48)\end{tabular}    \\
$|\Gamma|\ge 6$ & \begin{tabular}[c]{@{}c@{}}8.67\\ (0.48)\end{tabular} & \begin{tabular}[c]{@{}c@{}}8.67\\ (0.48)\end{tabular} & \begin{tabular}[c]{@{}c@{}}8.67\\ (0.48)\end{tabular} & 
\begin{tabular}[c]{@{}c@{}}0.30\\ (1.31)\end{tabular} \\
$|\Gamma|\ge 7$ & \begin{tabular}[c]{@{}c@{}}8.98\\ (0.15)\end{tabular} & \begin{tabular}[c]{@{}c@{}}8.98\\ (0.15)\end{tabular} & \begin{tabular}[c]{@{}c@{}}8.98\\ (0.15)\end{tabular} &
\begin{tabular}[c]{@{}c@{}}0.36\\ (1.68)\end{tabular}  \\
$|\Gamma|\ge 8$ & \begin{tabular}[c]{@{}c@{}}9.00\\ (0.00)\end{tabular} & \begin{tabular}[c]{@{}c@{}}9.00\\ (0.00)\end{tabular} & \begin{tabular}[c]{@{}c@{}}9.00\\ (0.00)\end{tabular} & 0                        \\ \bottomrule                                                           
\end{tabular}
\caption{
Approximating the revealed preference relation.
The left panel shows results in terms of the maximal lower contour sets.
The right panel shows the difference in the number of comparisons that are different.
Numbers are only computed for the subjects consistent with the given theory.
}

\label{tab:WelfareBounds}
\end{table}

\noindent
Table \ref{tab:WelfareBounds} presents the results on how well we can approximate the revealed preference relation using the bounds constructed.
The left panel presents the results in terms of the maximal lower contour set.
In particular, it presents  the average maximal size of the lower contour set with standard deviation in the parenthesis. 
It is evident that the largest difference between the two bounds is 0.33 which happens for  $|\Gamma|\ge 3$. Starting from $k=5$ the gap between the bound disappears ($\bar P = \ubar P$).
Let us note that the difference between $\ubar P$ and $P^*$ never exceeds $0.02$ on average.
This may be evidence that $\ubar P$ is a better approximation of the revealed preference relation ($P^*$).

The right panel shows the difference between $P^*$ and $\ubar P$ measured in terms of the number of comparisons that are present in $P^*$ but not in $\ubar P$, averaged across subjects. First, we see that the average number of different comparisons never exceeds 1. Moreover, we know that the difference only appears for 1 out of 113 subjects for the case of the generic CLA ($|\Gamma|>1$) and for 9 subjects out of 113 for when $|\Gamma|\ge 2$. For $3^{rd}$ and $4^{th}$-order theories the numbers account for 16\% and 14\% of consistent subjects correspondingly. Normalizing these numbers by the densities of the given preference relations we know that for the generic CLA it accounts for 4\% of the density and for the $2^{nd}$-order CLA it does for 3\% of the average density. For the $3^{rd}$ and $4^{th}$-order theories the difference accounts for 5\% and 3\% of the density.
Thus, we can overall conclude that $\ubar P$ is a good enough approximation of $P^*$ regardless of the metrics one wants to use.

\section{Conclusion}\label{sec:cnc}
In this paper, we provide a model of choice that enables an observer to draw conclusions about DM's preferences while assuming that the decision-maker suffers from inattention. We do this by imposing an \textit{attention floor} in the decision process; equivalently, assuming that DM makes at least one comparison in the process. We provide a characterization of our model based on both complete and incomplete data sets. Since we face the issue of multi-representation, the revealed preference relation has the following definition: $x$ is revealed better than $y$ if and only if $x$ is better than $y$ for every preference relation that could have generated the observed data. To evaluate the welfare relevance, we develop an algorithm to recover the revealed preference relation.


We apply our model to an experiment with an incomplete data set to evaluate the trade-off between the descriptive power and the welfare relevance of different attention floors.  
We show that the original CLA model does not help much to recover the revealed preference relation. 
Nonetheless, adding the minimal assumption that at least one comparison is made before finalizing decisions significantly improves the welfare relevance. 
In particular, this amended model allows us to recover about one-third of the preference relation that would be recovered under the full attention assumption. 
Interestingly, we find that this minimal assumption is almost costless in terms of descriptive power.  

Let us note that recovering preferences in this class of models is complicated due to their two-stage nature (shortlist-then-choose) that generates ``or-logic''. In such models revealed preferences could be of the form: $x$ is better than $y$ or $z$ is better than $w$. This type of revealed preference information is not restricted to the CLA model and can arise in other two-stage models \citep[see for instance][]{manzini2007sequentially,au2011sequentially,lleras2017more}. The efficacy of these models in explaining the observed data has been documented \citep[see][]{inoue2018limited}. In our opinion, a next important step is to investigate the welfare relevance of these models. 

\bibliographystyle{plainnat}
\bibliography{refs}

\clearpage
\appendix

\section{Proofs}

\subsection{Proof of Theorem \ref{thm:fulldata}}

\begin{proof}
($\Leftarrow$) Define $x \succ y \iff x = \C\{x,y\}$. As discussed before, SARP$^k$ implies no binary cycles on pairs. This implies that $\succ$ is a preference relation. Next, for $S$ define $\Gamma(S) = \C(S) \cup \{y \in S: x = \C\{x,y\}\}$. It follows from the definition that $\Gamma$ is an attention filter, and also that $\C(S) = \lmax{\succ}{\Gamma(S)}$ for all $S$. Hence, it only remains to prove that $|\Gamma(S)| \geq \min \{|S|,k\}$, for all $S$. Take $S: |S| \leq k$ and let $x = \C(S)$. We need to show $\Gamma(S) = S$. SARP$^k$ implies that $x \in \C \{x,y\}$ for all $y \in S$. Hence the desired conclusion follows. 

To finish the proof of this direction, take $S: |S|>k$. We need to show that $|\Gamma(S)| \geq k$. 
Assume by contradiction that $|\Gamma(S)| < k$. Since $x = \C(S) \in \Gamma(S)$ by definition, then let $\Gamma(S) = \{x, y_1,y_2,\ldots,y_n\}: n\leq k-2$. 
Note that $x = \C\{x,y_i\}$. Also let $\Gamma(S)\setminus S = A = \{z_1,z_2,\ldots, z_m\}$ be the set of elements DM does not pay attention to in $S$.
Note that $z_i = \C\{x,z_i\}$. 
Thus, by WARP(LA$^k$), we can remove $z_i$ from $S$, one by one, while preserving $x$ as choice. 
Thus, we conclude that $x = \C\{x,y_1,y_2, \ldots,y_n,z^*\}$, for $z^* \in \{z_1,z_2,\ldots,z_m\}$. 
Since $n \leq k-2$, then $|\{x,y_1,y_2,\ldots,y_n,z^*\}| \leq k$. Thus SARP$^k$ implies that $x \in \C \{x,z^*\}$ which is a contradiction. 

($\Rightarrow$) Assume that $\C$ is a $k^{th}$-order CLA with the preference relation $\succ$. To establish SARP$^k$, assume that $x$ is $k^{th}$-order indirectly chosen over $y$. Since DM's attention span is at least $k$, we conclude that $x \succ y$. Now consider a choice problem containing no more than $k$ elements and that includes both $x$ and $y$. Note that DM is fully attentive in such a problem and thus, since $x \succ y$, she will not choose $y$ over $x$. That is, $y$ is not $k^{th}$-order directly chosen over $x$. This establishes SARP$^k$.

To establish WARP(LA$^k$), take $x,y \in S \cap T$ with $|T| \leq k$, $x = \C (S)$ an $y = \C(T)$. Since DM is fully attentive in $T$, it follows that $y \succ x$. Since $x$ is chosen over $y$ in $S$-- an alternative that is preferred to $x$ -- it has to be the case that DM does not pay attention to $y$ in $S$; that is, $y \notin \Gamma(S)$. Since $\Gamma$ is an attention filter, then removing $y$ should not affect DM's attention set. That is, $\Gamma(S\setminus \{y\}) = \Gamma(S)$. Since $x$ was a maximum on $\Gamma(S)$, then it is also on $\Gamma(S\setminus \{y\})$. That is, $x = \C(S\setminus \{y\})$. This finishes the proof of this direction.
\end{proof}




\subsection{Proof of Theorem \ref{prop:RationalizabilityIncompleteExperiment}}
\begin{proof}
{\bf ($\Rightarrow$)}
Consider a rationalizable data set.
That is, there is an underlying preference relation $\succ$ and attention filter $\Gamma$ (where $\vert \Gamma(B) \vert \ge k$) such that
$$
\C(B) \succ y \text{ for every } y \in \Gamma(B).
$$
Let us show that the system \eqref{eq:MIP} is satisfied.
Given that $\succ$ is complete and transitive binary relation on no more than countable set of alternatives there is a utility function $u: X\rightarrow \R$ representing the $\succ$ (See \cite{rader1963existence}).
Moreover, this utility function is bounded.
Then, we can assign
$$
u_i = u(\C(B_i))
$$
for every $B_i\in \gB$.
Next, we define the $\delta_{ij}$ using the attention filter.
Let
\begin{equation*}
\delta_{ij} = 
    \begin{cases}
        1 & \text{ if } x_j \in \Gamma(B_i) \\
        0 & \text{ otherwise }
    \end{cases}
\end{equation*}
for every $x_j \in X$.
Then,
$$
\delta_{ij} = 0 \text{ for every } x_j \notin B_j
$$
is satisfied because $\Gamma(B_i)\subseteq B_i$, and 
$$
\sum\limits_{j\ne i} \delta_{ij} \ge \min\lbrace k,k_i\rbrace
$$
is satisfied by construction since $\vert \Gamma (B_i) \vert \ge  \min\lbrace k,k_i\rbrace$.
Next, we show that 
$$
u_i > u_j - (1-\delta_{ij})M  \text{ if } x_j \in B_i
$$
is satisfied
If $x_j \notin \Gamma(B_i)$, then $\delta_{ij} = 0$ and given that $u$ is bounded and $M\rightarrow \infty$.
If $x_j \in \Gamma(B_i)$, then $x_i=\C(B_i) \succ x_j$ (given that $(\succ,\Gamma)$ rationalizes the data set), and therefore $u_i> u_j$ since $u(x)$ represents $\succ$.
Finally, we show that 
$$
\sum\limits_{x_k \in B_i\setminus B_j}(\delta_{ik} + \sum\limits_{x_k \in B_j\setminus B_i}\delta_{jk} \ge 1 \text{ if } \C(B_i)\ne \C(B_j), \ \C(B_i),\C(B_j)\in B_i\cap B_j
$$
is satisfied.
Assume on the contrary that regardless 
$$
\C(B_i)\ne \C(B_j) \text{ and } \C(B_i),\C(B_j)\in B_i\cap B_j
$$
there $\Gamma(B_i)\cap (B_i\setminus B_j) = \emptyset$ and $\Gamma(B_j)\cap (B_j\setminus B_i) = \emptyset$.
Then, $\Gamma(B_i\cap B_j) = \Gamma(B_i)$ and $\Gamma(B_i\cap B_j) = \Gamma(B_j)$ since $\Gamma$ is an attention filter (removing element not in attention filter does not alter the filter).
Thus, we have $\Gamma(B_i) = \Gamma(B_j)$ and given that we have defined a strict preference relation equality of attention filters implies that $\C(B_i)=\C(B_j)$, that is a contradiction.

\bigskip

\noindent
{\bf ($\Leftarrow$)}
We start from showing that if there is a solution to system \eqref{eq:MIP}, then we can construct a revealed preference relation that is acyclic and thus can be completed.
Let $\hat \Gamma(B_i)$ be a temporary attention filter, that is $x\in \Gamma(B_i)$ if and only if $\delta_{ij} = 1$.
Let $\;\hat\succ\;$ be a revealed preference relation, that is, $x_i \;\;\hat\succ\;\; x_j$ if $x_i\in \C(B_i)$ and $x_j \in \hat\Gamma(B_i)$.

\medskip

\noindent
{\bf Claim 1.} If $x_i \;\;\hat\succ\;\; x_j$ then $u_i > u_j$.

\smallskip

\noindent
Assume that $x_i\;\hat\succ\; x_j$, then $x_i\in \C(B_i)$ and $x_j\in \hat\Gamma(B_i)$.
The latter implies that $\delta_{ij}=0$ by construction of $\hat\Gamma(B_i)$.
Then inequality $u_i>u_j - \delta_{ij}M$ has to hold as $u_i>u_j$.

\medskip

\noindent
Claim 1 links the constructed revealed preference relation to the solution of the system \eqref{eq:MIP}.
Let $\succ^{tc}$ be a transitive closure, that is $(x,y)\in \;\succ^{tc}$ if there is a sequence $x=x_1,\ldots,x_n=y$ such that $(x_j,x_{j+1})\in\; \succ$ for every $j\le n-1$.
In order to show that there is a complete extension of $\;\hat\succ\;$ it is enough to show that $\hat\succ^{tc}$ \emph{extends} $\;\hat\succ\;$.
That is, there are not $z,w\in X$ such that $(z,w)\in \;\hat\succ\;$ and $(w,z)\in \hat\succ^{tc}$.
Since $(w,z)\in \hat\succ^{tc}$, then there is a sequence $w=x_1,\ldots,x_n=z$ such that $(x_j,x_{j+1})\in \succ$ for every $j\le n-1$.
Then using Claim 1 we know that $u_j> u_{j+1}$ and given the transitivity of $>$ relation we know that $u_{x_1} > u_{x_n}$.
At the same time Claim 1 implies that $u_{x_n} > u_{x_1}$ since $z=x_n \;\hat\succ\; x_1 = w$.
That is a contradiction. 

\medskip

\noindent
Thus, there is a complete and transitive extension of $\;\hat\succ\;$, we denote by $\succ^*$.
Given that $\succ^*$ extends $\;\hat\succ\;$ we know that $x\;\hat \succ\; y$, then $x\succ; y$.
Thus, we are left to construct the attention filter.
If $B\in\gB$ let
$$
\Gamma(B) = \lbrace y\in B: \C(B)\succ^* y  \rbrace \cup \lbrace \C(B) \rbrace.
$$
Thus already constructed $\Gamma(B)$ allow to show that the data is rationalized, given that $\succ^*$ extends $\;\hat\succ\;$, since $\hat\Gamma(B)\subseteq \Gamma(B)$ for every $B\in \gB$.
By construction $\vert \hat \Gamma(B_i)\vert \ge \min\lbrace k_i,k\rbrace$, thus $\Gamma(B_i)\ge  \min\lbrace k_i,k\rbrace$.
Next, we extend the attention filter to $S\in 2^X\setminus \gB$.
Let 
$$
\hat \Gamma(S) = \bigcup_{B\in \gB: S\subseteq B; \Gamma(B)\subseteq S} \Gamma(B)
$$
be the candidate attention filter.
However, the candidate attention filter might be empty, thus we need to consider a different construct for that case.
Let 
\begin{equation*}
    \Gamma(S) = 
    \begin{cases}
    \hat\Gamma(S)   & \text{ if } \hat\Gamma(S)\ne \emptyset \\
    S               & \text{ otherwise }
    \end{cases}
\end{equation*}
for every $S\in 2^X \setminus \gB$.
Note that $\vert \Gamma(S) \vert \ge \min \lbrace k_i,k\rbrace$ in the case then it is the entire set (if $\vert\hat\Gamma(S)\vert=\emptyset$).
If $\vert\hat\Gamma(S)\vert\ne \emptyset$, then we need to consider two cases.

\medskip

\noindent
{\bf Case 1: $\vert S\vert \le k$.}
Suppose that $\vert \Gamma(S) \vert < k$.
If $\hat\Gamma(S)\ne \emptyset$, then there is $B\in \gB$ such that $\Gamma(B)\subseteq S\subseteq B$.
Note that $\Gamma(B)\ge \lbrace k_i,k\rbrace$ by construction, thus $\Gamma(B)=S=B$ that is a contradiction to $S\notin \gB$.

\medskip

\noindent
{\bf Case 2: $\vert S\vert > k$.}
Then for every $B\in \gB$ such that $S\subseteq B$ we know that $\vert B\vert >k$ and thus $\vert\Gamma(B)\vert \ge k$.
Hence, $\vert \Gamma(S) \vert \ge k$ since $\Gamma(B)\subseteq \Gamma(S)$.

\medskip
Finally, we are left to show that $\Gamma(S)$ is an \emph{attention filter}.
That is, if we take $S\in 2^X$ and $x\notin \Gamma(S)$, then $\Gamma(S\setminus\lbrace x\rbrace) = \Gamma (S)$.
Let $S'=S\setminus \lbrace x\rbrace$, for the further simplicity of the notation.
In order to finish the proof we need to consider the following four cases.

\medskip

\noindent
{\bf Case 1: $S\notin \gB; S'\notin \gB$.}
Take $y\in \Gamma(S)$, then there is $B\in \gB$ such that $\Gamma(B)\subseteq S\subseteq B$.
Given that $S'\subset S$, we know that $S'\subset B$.
Since $x\notin \Gamma(S)$, then $x\notin \Gamma(B)$, thus $\Gamma(B)\subseteq S'$ and thus $\Gamma(B)\subseteq \Gamma(S')$.
Take $y\in \Gamma(S')$,  then there is $B\in \gB$ such that $\Gamma(B)\subseteq S'\subseteq B$.
Assume on the contrary that $y\notin \Gamma(S)$, thus for all $B'\in \gB$ such that $\Gamma( B')\subseteq S\subseteq B'$, then $y\notin \Gamma(B')$.
By construction we know that $\C(B)\ne \C(B')$, then there is $z\in (\Gamma(B)\setminus B')\cup(\Gamma(B')\setminus B)$.
If $z\in \Gamma(B)\setminus B'$, then $z\in S'\subset S$, that is a contradiction to $S\subseteq B'$.
If $z\in \Gamma(B')\setminus B$, then $z\in S$ and $z\ne x$, thus $z\in S'$ that is contradiction to $z\subseteq B$.

\medskip

\noindent
{\bf Case 2: $S\in \gB$, $S'\in \gB$.}
By construction $y\in \Gamma(S)$ if and only if $\C(S)\succ^* y$.
Thus $\Gamma(S)=\Gamma(S')$, since $x\notin \Gamma(S)$.

\medskip

\noindent
{\bf Case 3: $S\in\gB$, $S'\notin \gB$.}
$x\notin \Gamma(S)$, then $\Gamma(S)\subseteq S'\subseteq S$.
Thus by construction $\Gamma(S)\subseteq \Gamma(S')$.
Next we show that $\Gamma(S')\subseteq \Gamma(S)$.
Take $y\in \Gamma(S')$ and assume that $y\notin \Gamma(S)$.
Then, there is $B\in \gB$ such that $\Gamma(B)\subseteq S'\subseteq B$ and $y\in \Gamma(B)$.
Thus, either $y\in \Gamma(S)$ that is a direct contradiction or $\C(B)\ne \C(S)$.
The latter implies that $\Gamma(B)\setminus S\ne\emptyset$, that is a contradiction to $\Gamma(B)\subseteq S'\subset S$.

\medskip

\noindent
{\bf Case 4: $S\notin \gB$, $S'\in \gB$.}
Take $y\in \Gamma(S)$, then there is $B\in\gB$ such that $\Gamma(B)\subseteq S\subseteq B$.
Thus, $\Gamma(B)\subseteq S'$ since $x\notin \Gamma(B)$ as $x\notin \Gamma(S)$.
Assume that $y\notin \Gamma(S')$, then $\C(S')\ne \C(B)$ and thus $\Gamma(B)\setminus S'\ne \emptyset$, that is a contradiction.
Take $y\in \Gamma(S')$ and assume that $y\notin \Gamma(S)$.
Then, for every $B\in\gB$ such that $\Gamma(B)\subseteq S\subseteq B$, $y\notin \Gamma(B)$.
Since $x\notin \Gamma(S)$, then $\Gamma(B)\subseteq S'$.
At the same time $y\in \Gamma(S')$ and $y\notin \Gamma(B)$ implies that $\C(B)\ne\C(S')$.
The latter implies that there is $z\in \Gamma(B)\setminus S'$ that is a contradiction to $\Gamma(B)\subseteq S'$.

\end{proof}

\subsection{Proof of Theorem \ref{thm:CorrectAlgorithm}}
Before we proceed with the proof we need to introduce some auxiliary notation.
Recall that we have defined $\cM(\gB,\C)$ as the set of pairs $(\succ, \Gamma)$ that rationalize the data set.
To proceed with the proof we define the
$$
\cM(\gB,\C,P) = \{(\succ, \Gamma) \in \cM:  P\subseteq \succ \}.
$$
That is a set of potential rationalizations of the data that extend a given binary relation $P$.

\begin{proof}
To complete the proof we need to state two supplementary claims.

\medskip

\noindent
\textbf{Claim 1:} $\ubar P \subseteq P^* \subseteq \bar P$.

\smallskip

\noindent
$\ubar P\subseteq P^*$ is simply by construction.
$P^* \subseteq \bar P$ is implied by the fact that $P^*\subseteq \succ$ for every $(\succ,\Gamma)\in \cM(\gB,\C)$.
Thus, $(x,y) \in P^*$ implies that $(x,y) \in \bar P$.

\bigskip

\noindent
\textbf{Claim 2:} $\cM(\gB,\C) = \cM(\gB,\C,P)$ for every $P\subseteq P^*$.

\smallskip

\noindent
First, we show that $\cM(\gB,\C) = \cM(\gB,\C,P^*)$.
By definition we know that $\cM(\gB,\C,P^*) \subseteq \cM(\gB,\C)$.
Thus, we are left to show that $\cM(\gB,\C) \subseteq \cM(\gB,\C,P^*)$.
Take $(\succ,\Gamma) \in \cM(\gB,\C)$, but by construction we know that $P^*\subseteq \succ$ and thus $(\succ,\Gamma) \in \cM(\gB,\C,P^*)$.
Now we consider the $P\subseteq P^*$.
We know that $\cM(\gB,\C,P) \subseteq \cM(\gB,\C)$, thus we are left to show that $\cM(\gB,\C) \subseteq \cM(\gB,\C,P)$.
We know that $\cM(\gB,\C) \subseteq \cM(\gB,\C,P^*)$ and $\cM(\gB,\C,P^*) \subseteq \cM(\gB,\C,P)$, then $\cM(\gB,\C) \subseteq \cM(\gB,\C,P)$.

\bigskip

\noindent
The first part of Claim 1 implies that we can initialize $P^*$ with $\ubar P$.
The second part of Claim 1 implies that it is enough to consider $(x,y) \in \triangle P$ as the potential candidates for $P^*$.
Thus, we are left to show that for $P\subseteq P^*$ if $\cM(\gB,\C)\neq \emptyset$ and $\cM(\gB,\C,P\cup \{(y,x)\}) =\emptyset$, then $(x,y)\in P^*$.
Suppose on the contrary $(x,y)\notin P^*$, then there is $y \succ x$ for $(\succ,\Gamma)\in \cM(\gB,\C) = \cM(\gB,\C,P)$.
Then, $P\cup \{(y,x)\} \subseteq \succ$ and thus $(\succ,\Gamma)\in \cM(\gB,\C,P\cup \{(y,x)\})$ that is a contradiction.

\end{proof}

\end{document}